\documentclass[12pt,reqno, times]{article}
\usepackage{mathptmx}

\usepackage[T1]{fontenc}
\usepackage[latin9]{inputenc}
\usepackage{geometry}
\usepackage{float}
\usepackage{amsmath}
\usepackage{amsthm}
\usepackage{graphicx}
\usepackage{setspace}
\usepackage[english]{babel}
\usepackage{color}
\usepackage[natbib, maxcitenames=5, mincitenames=1, style=apa]{biblatex}
\usepackage{csquotes}

\makeatletter

\theoremstyle{definition}
\newtheorem{defn}{\protect\definitionname}
\theoremstyle{plain}
\newtheorem{prop}{\protect\propositionname}
\theoremstyle{plain}
\newtheorem{lemma}{\protect\lemmaname}
\theoremstyle{plain}
\newtheorem{corollary}{\protect\corollaryname}
\makeatother

\usepackage{babel}
\providecommand{\definitionname}{Definition}
\providecommand{\propositionname}{Proposition}
\providecommand{\lemmaname}{Lemma}
\providecommand{\corollaryname}{Corollary}

\begin{document}
	\title{Coordinated Shirking in Technology Adoption}
	\author{Nicholas H. Tenev\thanks{Office of the Comptroller of the Currency. The views expressed in this paper are my own and do not represent the views of the OCC, the Department of the Treasury, or the United States government. This paper is the result of my independent research and has not been reviewed by the OCC. I am grateful to Enghin Atalay, Shane Auerbach, Alexander Clark, Dean Corbae, Steven Durlauf, Chao Fu, and Rasmus Lentz for helpful comments and suggestions. An earlier draft of this paper was circulated under the title ``Coordinated Shirking.''}}
	\maketitle
	\begin{abstract}
		This paper studies a model of technology adoption: a manager tries to induce a group of workers to exert costly effort to vet a new technology before they choose whether to use it. The manager finds it too costly to simultaneously replace large groups of unproductive workers, so they shirk when coordination is possible. Widely applicable technology expands productive possibilities but also provides an opportunity for coordinated shirking, and can thus lead to widespread production failure. Furthermore, even workers who learn that they are using flawed technology may continue to do so. Applications include mortgage securitization in the financial crisis of 2008, and the adoption of generative artificial intelligence.
	\end{abstract}
	
	Keywords: technology adoption, principal-agent, artificial intelligence, financial crisis, Gaussian copula, mortgage securitization \bigskip \\ 
	
	JEL codes: G01, O33, J33, M52
	%O33 	Technological Change: Choices and Consequences ? Diffusion Processes
	%G01 	Financial Crises 
	% J33 	Compensation Packages ? Payment Methods 
	% M52 	Compensation and Compensation Methods and Their Effects
	\section{Introduction}
	New technologies are typically expected to increase productivity and output. However, if not properly vetted they may be misused, with potentially ruinous consequences. For example, a formula used to price mortgage-backed securities could understate the risk of widespread defaults, as happened in the lead-up to the financial crisis of 2008. Or to give an example using an emerging technology, generative artificial intelligence (AI) can produce text that seems correct but actually contains serious factual errors (``hallucinations''). This paper shows that the very popularity of a technology can create an agency friction, leading to coordinated shirking amongst those who are supposed to be vetting it. Specifically, workers evaluating new technology may shirk their duty when it is sufficiently widespread in use, since their employers may not credibly be able to threaten to fire all of them. Counterintuitively, new technology can thus result in a drop in output even as it expands production possibilities. This paper models this agency friction, applies it to real-world examples, and discusses policy implications for mitigation.
	
	In the model, a lone principal (manager) seeks to induce effort from a group of agents (workers). A new technology has arrived, and is applicable to a certain fraction of the agents. If the technology is ``good'' it increases production for those who use it, but in rare cases it is ``bad'' and reduces production to zero. Agents with access to the new technology can exert costly research effort to obtain a noisy signal of whether it is good. After obtaining this signal (or shirking and seeing no signal), these agents decide whether or not to use the new technology. The principal is unable to observe the agents' effort or signals, but she can observe whether they use the technology and, eventually, their production. She is able to punish agents who fail to produce by replacing them, but some replacements cost more than others, so the total cost is convex in the number of agents punished. The main idea is that the principal may be unable to credibly threaten punishment to large groups of agents who shirk simultaneously: a ``too many to fail'' problem. The idea that mass punishment may be infeasible has been exploited in various contexts throughout history, such as ``round robin'' complaints that arranged petitioners' names in a circle to reduce the risk of individual reprisal, used by 18$^{\text{th}}$ century European sailors \citep{leeson2010rational} and Edo period farmers in Japan \citep{masamichi2021popular}, among others. It also appears in the work of \citet{acharya_cash---market_2008,acharya_too_2007} and \citet{farhi_collective_2012}, in which the promise of government bailouts in an aggregate downturn can prompt banks to correlate the returns of their investments. This paper's novelty is in the application of this mechanism to technology adoption, and focuses on the friction introduced by the interaction of technological progress with the convexity in the principal's cost of punishment. The focus here is on workers' technology adoption choices, which can interact with firm-level dynamics: hasty adoption by workers may exacerbate firms' rushing to establish first-mover advantage \citep{fudenberg1985preemption}, and undermine waiting for technology to be proven \citep{hoppe2000second}.
	
	Throughout, we apply the model to two important examples of technology adoption---one from recent history, and the other contemporary. 
	
	\bigskip
	\emph{1. Mortgage securitization and the 2008 financial crisis}
	
	The first application is the adoption of pricing models that facilitated the credit boom preceding the financial crisis of 2008. Gaussian copula models in particular were developed in the late 1950's but became popular with practioners in the early 2000's, and were ubiquitous in pricing collateralized debt obligation (CDO) including securities backed by subprime mortgages.\footnote{\citet{coval_economics_2009} provide a primer on securitized products and their role in the crisis; see also \citet{zimmer_role_2012}.} This technology became public knowledge in the finance community \citep{mackenzie_`formula_2012}\footnote{\citet{mackenzie_`formula_2012} study the use of Gaussian copula models and their role in the 2008 financial crisis, reporting on 95 interviews of financial services workers, of which 29 took place before July 2007.} and allowed risky loans to be bundled into seemingly safer securities, making them palatable to investors with less appetite for risk. But when more borrowers than predicted defaulted at once, the securitized products formed from their loans lost much of their value. A large asset price shock does not necessarily imply mistakes in modeling. However, \citet{zimmer_role_2012} finds that the degree to which house prices are related is insufficiently captured by a Gaussian copula, even using data available before 2002. And industry practitioners expressed ``pervasive dissatisfaction'' with Gaussian copula models even as early as 2006 \citep{mackenzie_`formula_2012}. So despite uncertainty and even pessimism about the use of a new financial technology, its use continued.\footnote{In fact, traders who wanted to use models other than Gaussian copulas had trouble convincing risk controllers that their positions were properly hedged \citep{mackenzie_`formula_2012}. So while using the new technology was essentially free, \emph{not} using it was made even more difficult by risk policies than this paper assumes.} But why? This paper provides an answer to this puzzle, modeling an agency friction that can cause a particularly useful innovation to precipitate a downturn in production. The model is simple, and cannot capture all the subtlety of model selection and evaluation, let alone the myriad other factors contributing to the financial crisis. Nonetheless, it cleanly illuminates an important tension: a technology's popularity can belie its unsoundness.  
	
	\bigskip
	\emph{2. Generative artificial intelligence}
	
	This friction is also useful in understanding risk in the adoption of generative artificial intelligence (AI). In response to a user prompt written in plain language, generative AI such as large language models (LLMs) can quickly produce responses that read as though written by a human, accomplishing tasks such as summarizing documents \citep{zhang2024benchmarking} and generating computer code \citep{idrisov2024}. Designed to mimic human intelligence, the range of industries and activities this technology may improve is vast, if still limited currently \citep{mcelheran2024ai}, and informal adoption is growing rapidly \citep{bick2024rapid}.\footnote{While this paper focuses on AI as a tool to enhance existing productive processes, it may also replace workers \citep{hui2024short}.} Despite generating responses that are facially plausible, one known issue with this technology, sometimes called ``hallucination,'' is that the output can include text that facially seems reasonable but actually includes factual inaccuracies (summarized by e.g. \cite{ji2023survey}). In one well-publicized example, a legal brief created using generative AI included quotations and citations of legal opinions that did not exist \citep{weiser2023}. Given the rapid pace at which AI is being both developed and adopted, hallucination may not turn out to be the greatest risk it poses. Nonetheless, it provides another good example of how new technology can reduce productivity if improperly vetted or understood. This paper will show how the widespread applicability of generative AI may increase the risk that it is understudied and thus misused. Applying the paper's model to an emerging technology entails more speculation, but in return can provide insights that are more useful in the present.\footnote{Other factors moderating AI adoption range from the task structure of a firm \citep{acemoglu2022artificial} to local social trust \citep{tubadji2021cultural}.}
	
	In the model, even workers who find out that the  technology is flawed may continue to use it, since they do not fear punishment for doing so. This behavior has a flavor of herding; \citet{devenow_rational_1996} summarize papers on herding in financial markets.\footnote{The mechanism is also similar to herding in crime, where the presence of other criminals encourages criminal activity by reducing the chance of punishment \citep{dipasquale_angeles_1998}. This paper's relative contribution is to highlight the significance of agent anonymity, and the application to technology adoption.} Others (e.g. \citet{acharya_causes_2009} and \citet{rajan_bankers_2008}) have noted that analysts' compensation structures gave little incentive to care about the long-term performance of their work (a theme repeated here), but do not explain why reputation concerns or the threat of firing were not more effective restraints. This paper specifically explains why the threat of firing workers for poor performance may have been inadequate in these cases. Furthermore, it explicitly links a production downturn to the advent of a new technology, potentially helping to explain the timing of the financial crisis.
	
	It also illuminates several possible solutions to the threat of coordinated shirking. First, expectations are important: if workers for some reason expect that nobody else is shirking, they too will be diligent. Second, shifting worker compensation from variable (bonuses) to fixed (salary) may exacerbate shirking, as it uncouples workers' incentives from their output. Finally, the problem of coordinated shirking can be resolved if there is a commonly known complete strict ordering of workers (such as seniority) that the firm can use as a basis for punishment. The most senior person of any group of prospective shirkers will be motivated to exert effort, unraveling the shirking behavior of the group. 
	
	Section \ref{sec:Model} introduces the model. Section \ref{sec:Results} derives conditions for equilibria and calculates production loss compared to the first-best case. Finally, Section \ref{sec:Applications-and-discussion} applies the model to the financial crisis and AI adoption and discusses policy implications. Thoughout, proofs and derivations are deferred to the Appendix.
	
	\section{Model\label{sec:Model}}
	
	There exists a continuum of risk-neutral agents of measure one,\footnote{Using a continuum of agents rather than a finite set of discrete agents eases the analysis, but is not qualitatively important.} indexed by $i\in\left[0,1\right]$, and one risk-neutral principal. The agent maximizes expected wages and the principal maximizes expected output less the summed wages paid to the agents. The timing is a one-shot model of research, technology adoption, and production, detailed below.
	%They share a common time discount rate $\delta\in\left(0,1\right)$. 
	\subsection{Production technology}
	
	At the start of the model, new technology arrives for some workers: specifically, agents with index $i\leq h$ exogenously receive access. The parameter $h$ thus determines how widely available the technology is. In the context of mortgage securitization, $h$ might be the fraction of workers at a given firm for whom a new pricing model is at least facially useful; in the example of AI adoption, it might be the fraction of workers with access to a particular LLM. The effect of $h$ on equilibria is the focus of this paper. 
	
	The new technology is `good' with probability $\pi$. A good new technology yields production of $1+g$, but a bad new technology yields production of $0$. If an agent elects not to use the new technology (or it is not available to that agent), production is 1. An individual agent $i$'s production is thus
	\begin{equation}\label{eq:production}
		f_{i} = \begin{cases} 
			1
			& \text{if new technology is not used;}\\
			0 
			& \text{if new technology is available, used, and bad;} \\
			1+g 
			& \text{if new technology is available, used, and good.} \\
		\end{cases}
	\end{equation}
	
	The principal owns the output of production. 
	
	\subsection{Research effort and shirking}
	
	When given the opportunity to use a new technology ($i\leq h$), agents cannot directly observe whether it is good. However, they can each undertake costly effort to obtain a noisy, binary signal. Specifically, by paying a cost $c\geq0$, an agent learns the value of $\eta$, where $\eta=1$ indicates that the new technology is good and $\eta=0$ indicates that it is not. The chance that the signal is wrong (regardless of what it says) is $\epsilon\in\left[0,\frac{1}{2}\right]$. Forgoing effort will, as is usual in the literature, be called `shirking.' Assume that the signal is worth heeding and that research is socially optimal---this is specified in detail later in Equations \ref{eq:growth} and \ref{eq:cost}.
	
	\subsection{The principal}
	
	Before production is realized, the principal pays $w$ to each agent who used the new technology---a higher wage for higher expected productivity. The principal then receives the fruits of production, $f=\int_{0}^{1}f_{i}\text{d}i.$
	
	The principal cannot observe each agent's effort, but can eventually observe
	their output, and decide whom to fire and replace. There may be some workers who are nearby and qualified for the job whom the principal can use to replace staff at low cost. Replacing large numbers of workers, though, may oblige the firm to pay relocation or retraining costs for some of the new hires. Assume that firing workers is costless, and a unit mass of potential replacements exists. The principal can observe the cost of hiring each potential replacement $q\left(z\right)$, and the total cost of using the set of replacements $Z\subset\left[0,1\right]$ is $\int_Z q\left(z\right)\text{d}z$.
	
	The principal does not have to punish all agents who fail to produce, but rather can punish failure at a rate $\gamma$, which may depend on the number of potential failures $h$, such that an agent whose production fails is replaced with probability $\gamma\left(h \right)$. We allow the principal to commit to a strategy, which is necessary to incentivize the principal to bear the cost of punishment in this one-shot setup. However, this is just for simplicity---the game is easily extended to a repeated version in which commitment is not necessary.\footnote{An earlier draft of this paper, \citet{tenev2018coordinated}, studies a repeated game version, with similar results.} Focusing on the case with commitment also emphasizes that the driving friction here is not a commitment problem but rather the convexity of the principal's punishment function.
	
	\subsection{Timing}
	
	The timing of the model is as follows:
	\begin{enumerate}
		\item The principal commits to a firing policy $\gamma\left(h\right)$.
		\item A technology state $h$ is realized, and observed by all.
		\item Agents with access to the new technology ($i\leq h$) choose whether to exert research effort or to shirk.
		\item Each agent who chose effort pays a cost $c$ and receives a private binary signal $\eta$ of whether the new technology is good.
		\item Agents with access to the new technology choose whether or not to use it.
		\item The principal pays agents who used the new technology $w$. 
		\item Production $f_{i}$ is realized for each agent (Eq. \ref{eq:production}), and observed by the principal.
		\item A fraction $\gamma\left( h\right)$ of the agents who failed to produce are fired. The principal chooses their replacements $Z$, and pays the cost of replacement $\int_Zq\left(z\right)\text{d}z$.
		\item Agents who were not replaced receive continuation value $v_c$. 
	\end{enumerate}
	In the worker-firm context, this is equivalent to workers receiving a salary increase of $w$ for increased prospective productivity. The mechanism underlying this wage structure is not crucial to the results, so it is left unspecified (and the wage agents are paid when they don't use the new technology is normalized to zero). The continuation value $v_c$ can be interpreted as expected future wages from employment in the future, and is assumed to be high enough that the agent can be induced to exert effort through the threat of replacement, detailed below in Equation \ref{eq:vcont}.\footnote{$v_c$ does not depend on whether the agent uses the new technology now. Once production is realized it will be known if the new technology is good, and all agents with access would use it in subsequent periods if it is.} 
	
	In the context of pricing mortgage-backed securities, using the same financial technology made it easier for traders to claim the full expected future profits of such deals as justification for (current) bonuses, since the deals were priced using standard formulae \citep{mackenzie_`formula_2012}. This corresponds well to the timing of the model, in which workers receive payment for prospective output before output is actually realized. The case of compensation based on performance is discussed in \ref{subsec:The-financial-crisis}.
	
	\subsection{Strategies}
	
	Each agent with access to the new technology must choose effort or shirking, which may depend on how widely useful the technology is ($h$). For each realization of signal $\eta$ as well as a null signal (if the agent shirks), the agent must choose whether to use the new technology---a choice that may also depend on $h$. 
	
	%So an agent's strategy is a map from $\left[0,1\right]$ (the support of $\mu_{t}$) to $\left\{ \text{E},\text{S}\right\} \times\left\{ 0,1\right\} ^{3}$.
	
	Given that a fraction $y$ of agents produced successfully, the principal's strategy consists of a rate $\gamma\left(y\right):\left[0,1\right]\rightarrow\left[0,1\right]$ which determines the fraction of agents whose production failed to punish by replacing them. If all agents pursue the same strategy, this is equivalent to writing $\gamma\left(h\right)$, since either $h$ agents will fail or none will. We focus on Nash equilibrium as the solution concept.
	
	\section{Results\label{sec:Results}}
	
	\subsection{First-best}
	
	Fix a new technology $h$ and focus on the problem of an agent who has access to it. 
	
	\begin{lemma}\label{lem:efficientsignal}It will be efficient for agents who exert effort to act according to the signal received if 
		
		\begin{equation}
			g\in\left(\frac{\epsilon}{1-\epsilon},\frac{1-\epsilon}{\epsilon}\right).\label{eq:growth}
		\end{equation}
	\end{lemma}
	Equation \ref{eq:growth} describes a new technology that can increase productivity enough ($g$) to justify trusting the signal, but not so much that it should be used regardless of what the signal says.
	
	\begin{lemma}\label{lem:cost}
		Exerting effort to obtain the signal is efficient if  
		\begin{equation}
			c<\left(1-\pi\right)\left(1-\epsilon\right)-\pi\epsilon g.\label{eq:cost}
		\end{equation}
	\end{lemma}
	In other words, the signal cost ($c$) is less than the productivity gained by avoiding bad technology minus the productivity missed by forgoing good new technology owing to incorrect signals.  For the remainder of the paper, assume that \ref{eq:growth} and \ref{eq:cost} hold: researching the new technology is socially optimal. 
	
	If all agents were exerting effort, any failures in production would owe to chance. In the first-best case all agents would exert effort (since it is efficient) and the principal would never replace them, since it is costly. Given technology $h$, expected aggregate production in the first-best case is
	\begin{equation}
		\int_{0}^{1}f_{i}\text{d}i=\left(1-h\right)+h\left[\left(1+\pi g\right)\left(1-\epsilon\right)+\pi\epsilon\right].\label{eq:first best}
	\end{equation}
	
	Total welfare is simply output less agent effort costs, since wages and the continuation value are a transfer. Total welfare in the first-best case is $\left(1-h\right)+h\left[\left(1+\pi g\right)\left(1-\epsilon\right)+\pi\epsilon-c\right].$

	\subsection{Equilibria}
	
	As in other principal-agent settings, it is a Nash equilibrium here for agents to always shirk and for the principal to never punish (i.e. fire and replace them). However, the more interesting question is what conditions will induce effort, as it is socially optimal. The rest of the paper focuses on equilibrium with effort.
	
	\subsubsection{The principal's problem}
	
	If the principal chooses to replace some agents, she will optimally use the least costly replacements first. The lowest cost of replacing a fraction $x\in\left[0,1\right]$ of its workers is $r\left(x\right) := \text{min}_{\rho\subset\left[0,1\right]}\int_\rho q\left(z\right)\text{d}z$ such that the measure  $\lambda\left(\rho\right)=x$. If the principal uses the least costly replacements first, then increasing the measure of agents replaced requires using increasingly costly replacements. 
	\begin{lemma}\label{lem:convex} The total cost of replacement $r\left(x\right)$ is convex in the measure $x$ of agents replaced.\end{lemma}
	
	The rest of the principal's strategy consists of a policy $\gamma\left( h \right):\left[0,1\right]\rightarrow\left[0,1\right]$:
	given the fraction of projects that failed, she has to decide what
	fraction of the failures to punish/replace ($\gamma$). Assume that punishing at rate $\gamma\geq\underbar{\ensuremath{\gamma}}$ will induce the agent to exert effort (conditions for this to hold will be derived in \ref{subsec:The-agent's-problem}). Since punishment is costly, the principal will in all cases choose either $\gamma=0$ or $\gamma=\underbar{\ensuremath{\gamma}}$; anything between these two would be costly without inducing effort, while anything greater than $\underbar{\ensuremath{\gamma}}$ would be effective but unnecessarily costly. The question then becomes when the principal should punish and when she should not.
	
	Given a policy $\gamma\left(\bar{f}\right)$, let $\Gamma$ be the subset of the unit interval for which the principal punishes agents: $
	\Gamma\equiv\left\{ h:\gamma\left( h \right)=\underbar{\ensuremath{\gamma}}\right\}$. For $h\in\Gamma$, agents know that they will be punished for failure with probability $\underbar{\ensuremath{\gamma}}$, and are thus induced to exert effort in these states. 
	
	\begin{lemma}\label{lem:punish}
		In equilibrium, the principal will punish agents only in states $h$ such that
		\begin{equation}
			\frac{\left(1-\pi\right)\left(1-\epsilon\right)-\pi\epsilon g}{\left(1-\pi\right)\epsilon}h \geq r\left(\underbar{\ensuremath{\gamma}}h\right).\label{eq:IRcommit}
		\end{equation}
	\end{lemma}
	This characterizes the technology states in which the manager will want to commit to firing and replacing workers.
	
	\begin{defn}
		Let $\Gamma^{\star}$ be the set of all $h\in\left[0,1\right]$ that satisfy Equation \ref{eq:IRcommit}---the set of $h$ for which the principal can credibly threaten punishment under commitment.
	\end{defn}
	
	\begin{prop}\label{prop:threshold}
		In an equilibrium with nonzero punishment, the principal's optimal punishment policy $\gamma\left(h\right)$ will be a threshold rule:
		\begin{equation}
			\gamma^{\star}\left(h\right)=\begin{cases}
				\underbar{\ensuremath{\gamma}} & h\in\left[0,\tilde{h}\right)\\
				\gamma\in\{0,\underbar{\ensuremath{\gamma}}\}  & h=\tilde{h}\\
				0 & h\in\left(\tilde{h},1\right]
			\end{cases}.\label{eq:threshold}
		\end{equation}
		
		(The optimal policy at exactly the threshold $\tilde{h}$ can be $\underbar{\ensuremath{\gamma}}$ only if $\tilde{h}$ satisfies Equation \ref{eq:IRcommit}; otherwise, it will be 0.) 
	\end{prop}
	
	If $r\left(\cdot\right)$ is differentiable at zero, a sufficient condition for the set of states in which punishment is feasible $\Gamma^{\star}$ to be non-empty is $\frac{\left(1-\pi\right)\left(1-\epsilon\right)-\pi\epsilon g}{\left(1-\pi\right)\epsilon \underbar{\ensuremath{\gamma}}}>r^{\prime}\left(0\right).$ Otherwise, Equation \ref{eq:IRcommit} would hold with equality only at $h=0$, after which the costs of punishment would outweigh the benefits. 
	
	\subsubsection{The agent's problem\label{subsec:The-agent's-problem}}
	
	Assume the chance of punishment is $\gamma\left(h\right)$ for agents who fail to produce, which may depend on the technology state $h$, as this governs the maximum number of agents whose production can fail in the period. An agent with access to the new technology will exert effort if the expected value of doing so exceeds the expected value of shirking. 
	
	\begin{lemma}\label{lem:effort}
		Given chance of punishment $\gamma\left(h\right)$ for failed production, the agent will exert effort only if
		\begin{equation}
			\left(1-\pi\right)\left(1-\epsilon\right) \gamma\left(h\right)v_c \geq c + \left(1-\pi\left(1 - \epsilon\right) - \left(1 - \pi\right)\epsilon\right)w  . \label{eq:effort}
		\end{equation}
	\end{lemma}
	That is, the increased chance of continuation in the game {must outweigh the cost of effort and forgone wages owing to negative signals.
		
		\begin{defn}
			Let $\underbar{\ensuremath{\gamma}}$ be the lowest possible punishment rate which induces effort (Equation \ref{eq:effort} holds): 
			\begin{equation}
				\underbar{\ensuremath{\gamma}} = \frac{c + \left(1-\pi\left(1 - \epsilon\right) - \left(1 - \pi\right)\epsilon\right)w}{\left(1 - \pi\right)\left(1 - \epsilon\right) v_c}.\label{eq:gamma}
			\end{equation}
			
		\end{defn}
		
		It is obvious from Equation \ref{eq:effort} that if costs and wages are positive, the value of the signal alone is not enough to induce effort in the absence of punishment ($\gamma=0$). That is, only the threat of punishment can induce effort. To further ensure this is an interesting problem, assume parameters and a wage incentive $w$ such that if failure is always punished ($\gamma=1$), effort is induced: 
		\begin{equation}\label{eq:vcont}
			v_c \geq \frac{c+\left(1-\pi\left( 1-\epsilon\right)-\left(1-\pi\right)\epsilon\right)w}{\left(1-\pi\right)\left(1-\epsilon\right)}
		\end{equation}
		This ensures that $\underbar{\ensuremath{\gamma}} \in \left[0,1\right]$. In other words, assume that the noisiness of the signal is not so high as to prevent the principal from being able to induce effort. 
		
		An equilibrium with effort is then simply characterized by $\left\{ \underbar{\ensuremath{\gamma}},\tilde{h}\right\} $. 
		Here $\tilde{h}$ (given by Equation \ref{eq:IRcommit}) is the maximum fraction of agents the principal is willing to punish, and characterizes the principal's policy rule $\gamma\left(h\right)$ of the form given by Equation \ref{eq:threshold}. Below this threshold, the fraction of unproductive agents punished $\underbar{\ensuremath{\gamma}}$ is characterized by Equation \ref{eq:effort}. \textcolor{black}{But if the technology is applicable to more than $\tilde{h}$ agents, the principal is unable to induce effort.}

		\begin{corollary}\label{cor:drop}
			Expected production and welfare drop as the measure of agents the technology applies to surpasses the maximum the principal is willing to punish, $\tilde{h}$.
		\end{corollary}

		The model thus has a surprising prediction: a technological advance that is useful to too many people may actually trigger a downturn in production. This is the main result of the paper.
		
		\section{Applications and discussion\label{sec:Applications-and-discussion}}
		
		\subsection{Firms, workers, and aggregate downturns\label{subsec:The-financial-crisis}}
		We can now consider how the results apply to various contexts. While not modeled explicitly here, a financial technology innovation such as a pricing model may work fine in some macroeconomic states (say, low interest rates) but not in others. If it is useful to a large fraction of workers, the firm is unable to credibly threaten to fire workers who fail to produce because the cost of replacing them is too high. Knowing this, workers use the technology without researching it first, as the firm can't fire everyone. Production in this example will be high in states where the technology is good but will drop in states where the technology fails. The model can thus link the failure of models used to price securitized products to their popularity. 
		
		A similar concern is that a technology may be good for some use cases but bad for others. So a firm that has success using a large language model to generate code, which can quickly be tested and verified, may also start using for summarizing complex legal documents before it is proven capable of doing so reliably. While generative AI has yet to produce examples of large-scale failure similar to the 2008 financial crisis, its growing popularity coupled with difficulty in assessing its output \citep{mcintosh2024inadequacies} signals that coordinated shirking of pre-implementation evaluation may indeed be a risk in AI adoption.
		
		It is important to note that opportunities for coordinated shirking (i.e. $h_{t}>\tilde{h}$) may arise rarely. Firm management may think that their existing compensation and governance policies are sufficient to address any agency frictions simply because technology that is both widely useful (high enough $h_{t}$) but flawed for its current use has not come along in a while. Furthermore, consider a worker who somehow learns that a new, widely useful technology is flawed. Even with this knowledge, he will still choose to use the new technology if everyone else is, as it nets him a higher wage, and there is no chance of punishment. Research effort in this model has the flavor of a public good, since the technology is common. If it were possible for just one worker to exert the research effort and then share the signal with the others before technology choices were made, that would fix the problem. But once everyone is using a technology, it may be difficult to dissuade them, even if its flaws are made known. This may explain why even as doubts regarding the soundness of the credit boom arose, industry practices were slow to change.
		
		Of course, many financial services employees were indeed fired or laid off during the financial crisis. This paper's model predicts total amnesty for large enough groups of shirkers, but it is simplistic in that the firm size remains the same---if the firm size shrank during recessions, for example, there could be layoffs during a downturn. And since workers would not individually be able to affect the probability
		of a downturn, the core mechanism of the model would remain intact. A quant choosing to use industry-standard pricing formulae could reasonably expect this choice not to affect his future employability, even if the technology turned out to be unsound. Similarly, workers in various industries may see little personal downside to using generative AI in their work, if everyone is doing it.
		
		\subsection{Policy implications}
		
		Following the financial crisis of 2008, capping bankers' bonuses was widely discussed as a means of controlling risk-taking. In 2013, the European Union restricted banker bonuses to 100\% of pay, or up to 200\% with shareholder approval \citeyearpar[PwC LLP,][]{pricewaterhousecoopers_llp_eu_2013}. This shifts the balance of bankers' compensation from variable (bonus) to fixed (salary). Can such a policy reduce risk-taking?
		
		In this paper, the vital friction arises precisely because firms cannot punish workers by paying them less, only by firing them. Restricting variable compensation may thus exacerbate this friction, and increase the chance of a downturn. Variable compensation aligns workers' incentives with the returns of their projects. If in the model firms were allowed to pay workers their realized production rather than their potential production (similar to end-of-year bonuses), workers would be motivated to put forth research effort all the time, and the friction would disappear.
		
		Of course, timing is key. If the outcomes of workers' projects are not realized for years after they are paid, then end-of-year variable compensation would be no different from fixed compensation during the year. In this case, something stronger would be needed, such as clawbacks (in which the firm is able to reclaim bonuses if long-term performance objectives are not met) or payment in stock or stock options
		(which would lose value if the projects failed). These prescriptions are not new, but this paper illuminates another reason they may be important: the inadequacy of firing workers as a way to induce effort when workers are able to coordinate their failures.
		
		Finally, note that coordinated shirking relies on anonymity. Another potential solution is thus to use a publicly known order of punishment in the event of production failure. An analogous mechanism applies to speeding on the highway: if the traffic cop can only pull over one car but will pull over the first car in any group of speeders, then this will unravel commuters' incentive to speed in a large group. In the context of technology adoption at work, if firms are able to somehow identify a salient individual for any possible subset of workers and commit to punishing that individual, coordinated shirking can be avoided. For example, the firm could commit to firing the most senior worker (in terms of tenure at the firm) in a group whenever their production fails. Of course, this may not always be possible, or enough. Workers may know that a lot of other people are using a new technology but remain unsure exactly who. Or a group of workers may have all started at the same time, making seniority an insufficiently distinguishing characteristic. Nonetheless, to the extent that firms can create a complete, strict ordering of workers that is known by all, it may be able to prevent instances of coordinated shirking.\footnote{Technically, all that is needed is a map $f\left(\cdot\right)$ from the power set of agents to agents such that $f\left(a\right)\in a$; a complete strict order is simply an easy example.}
		
		\printbibliography
		
		\section*{Appendix: Proofs}
		\begin{proof}[Proof of Lemma \ref{lem:efficientsignal}]
			By ignoring the new technology, the agent is guaranteed production of 1. The expected production if the agent does not exert research effort (shirks) but uses the new technology anyway is \textcolor{black}{ $\pi\left(1+g\right).$} If the agent does exert research effort and receives a favorable signal ($\eta=1$), expected production using the new technology is $\left(1-\epsilon\right)\left(1+g\right).$ This is greater than the production of 1 guaranteed by forgoing the new technology if $g>\frac{\epsilon}{1-\epsilon}.$ If the agent exerts research effort and receives an unfavorable signal ($\eta_{i}=0$), expected production using the new technology is $\epsilon\left(1+g\right).$ Trusting the signal and forgoing the new technology will be preferable if $g<\frac{1-\epsilon}{\epsilon}.$ 
		\end{proof}
		\begin{proof}[Proof of Lemma \ref{lem:cost}]
			Expected production when exerting effort is 
			$\pi\epsilon+\left(1-\pi\right)\left(1-\epsilon\right)+\pi\left(1-\epsilon\right)\left(1+g\right)$, while expected production without effort is $\pi\left(1+g\right).$ Taking the difference and comparing it to the cost of effort yields the result.
		\end{proof}
		
		\begin{proof}[Proof of Lemma \ref{lem:convex}] Suppose the cost $r\left(\cdot\right)$ of replacing a given number of agents exhibits concavity somewhere: there exist $x<x^\prime<x^{\prime\prime}$ (with associated minimizing subsets of agents $\rho$, $\rho^\prime$, $\rho^{\prime\prime}$) such that $\frac{r\left(x^\prime\right)-r\left(x\right)}{x^\prime-x} > \frac{r\left(x^{\prime\prime}\right)-r\left(x^{\prime}\right)}{x^{\prime\prime}-x^\prime}$. In other words, the additional replacements as the number of agents replaced increases from $x$ to $x^\prime$ are more costly on average than the additional replacements going from $x^\prime$ to $x^{\prime\prime}$. Assume that $\rho\subset\rho^\prime\subset\rho^{\prime\prime}$---this is without loss, since if it is not true (because some agents cost the same to replace) it will be true for three sets of equally costly agents. Let $y:=\min \{x^{\prime\prime}-x^\prime , x^{\prime}-x\}$. Then, consider forming an alternative to $\rho^\prime$ by taking the least costly $y$ agents from $\rho^{\prime\prime}\setminus\rho^\prime$ and swapping them with the most costly $y$ agents from $\rho^\prime\setminus\rho$. Because the average cost of $\rho^{\prime\prime}\setminus\rho^\prime$ is strictly less than the average cost of $\rho^\prime\setminus\rho$, this yields a set of agents of measure $x^\prime$ that is less costly to replace than $r\left(x^\prime\right)$. This violates the definition of $r\left(\cdot\right)$ as the least cost of replacing a given measure of agents, and demonstrates that $r\left(\cdot\right)$ cannot exhibit any concavity.
		\end{proof}

		\begin{proof}[Proof of Lemma \ref{lem:punish}]
			The principal's value of employing policy $\gamma\left(\cdot\right)$ given a technology state $h$ for which the principal does punish is: $\left(1-h\right)+h\left[\pi\left(1-\epsilon\right)\left(1+g\right)+\pi\epsilon+\left(1-\pi\right)\left(1-\epsilon\right)\right]	-\left(1-\pi\right)\epsilon r\left(h\gamma\left(h\right)\right)-wh.
			$ The principal's value of employing policy $\gamma\left(\cdot\right)$ given a technology state $h$ for which the principal does not punish is: $\left(1-h\right)+h\left(\pi\left(1+g\right)-w\right).$	The difference is $
			h\left[-\pi\epsilon g+\left(1-\pi\right)\left(1-\epsilon\right)\right]-\left(1-\pi\right)\epsilon r\left(\gamma\left(h\right)h\right). $	For a policy $\gamma\left(\cdot\right)$ to be optimal, this must be greater than zero. Using the fact that in equilibrium the principal punishes at rate $\underbar{\ensuremath{\gamma}}$, we have the result.
		\end{proof}
		
		\begin{proof}[Proof of Proposition \ref{prop:threshold}]
			Since by Lemma \ref{lem:convex} $r\left(\cdot\right)$
			is convex, $r\left(\underbar{\ensuremath{\gamma}}h\right)$ is convex
			in $h$ given any $\underbar{\ensuremath{\gamma}}\in\left[0,1\right]$.
			Recalling that $r\left(0\right)=0$, this implies that given $\underbar{\ensuremath{\gamma}}$, if $h\in\left[0,1\right]$ satisfies Equation \ref{eq:IRcommit}, then so does any $h^{\prime}$ satisfying $0<h^{\prime}<h$. So letting $\tilde{h}$ be the supremum of the set of values in $\left[0,1\right]$ for which Equation \ref{eq:IRcommit} holds, the principal can and will optimally punish for all $h$ up to the threshold $\tilde{h}$. As discussed before, the principal will in all cases choose either $\gamma=0$ or $\gamma=\underbar{\ensuremath{\gamma}}$. So under commitment the principal's optimal policy is the threshold rule given by Equation \ref{eq:threshold}. The optimal policy at $\tilde{h}$ can be 0 for any parameters, or $\gamma^{\star}\left(\tilde{h}\right)=\underbar{\ensuremath{\gamma}}$ only if $\tilde{h}$ satisfies Equation \ref{eq:IRcommit} (in which case two threshold equilibria exist).
			
		\end{proof}
		
		\begin{proof}[Proof of Lemma \ref{lem:effort}]
			If the agent does exert effort, the expected payoff is
			\begin{equation}
				- c + \left(\pi\left(1-\epsilon\right)+\left(1-\pi\right)\epsilon\right)w + \left(1-\left(1-\pi\right)\epsilon\gamma\left(h\right)\right)v_c.
				\label{eq:veffort}
			\end{equation}
			
			The agent incurs cost $c$ for exerting effort. If the agent receives a signal that the technology is good (chance $\pi\left(1-\epsilon\right)+\left(1-\pi\right)\epsilon$) he will use it, and receive an additional $w$ for doing so. Only if the technology is bad and the signal incorrect is punishment triggered, resulting in a $\gamma\left(h\right)$ chance of the agent being replaced and losing the continuation value $v_c$.

			The agent's expected value of shirking on research but using the new technology anyway is
			\begin{equation}
				w + \left(1 - \left(1 - \pi\right)\gamma\left(h\right)\right)v_c.
				\label{eq:vshirk}
			\end{equation}

			The agent receives the increased wage $w$ for using the new technology, but if the new technology is bad (chance $1-\pi$) and he is selected for punishment (chance $\gamma\left(h\right)$), he misses out on the continuation value $v_c$. Taking the difference between expressions \ref{eq:veffort} and \ref{eq:vshirk} yields the result.
		\end{proof}
		\begin{proof}[Proof of Corollary \ref{cor:drop}]
			With the threat of punishment, expected output is $\left(1-h\right)+h\left[\left(1+\pi g\right)\left(1-\epsilon\right)+\pi\epsilon\right]$. Without it, expected output is $\left(1-h\right)+h\pi\left(1+g\right)$. The difference is $h\left[\left(1-\epsilon\right)\left(1-\pi\right)-\pi g\epsilon\right]$. By Equation \ref{eq:cost} this is greater than $c$, and since $c$ is assumed to be positive, so is this equilibrium drop in output as $h$ surpasses $\bar{h}$. Similarly, the expected welfare loss at $\tilde{h}$ (compared to the first best, Equation \ref{eq:first best}) is $h\left[\left(1-\epsilon\right)\left(1-\pi\right)-\pi g\epsilon+c\right]$. By Equations \ref{eq:growth} and \ref{eq:cost}, this is positive. 
		\end{proof}

	\end{document}